\documentclass[11pt]{article}
\pdfoutput=1 

\usepackage[utf8]{inputenc}
\usepackage[margin=1in]{geometry}
\usepackage{amsfonts,amsmath,amssymb,amsthm,mathtools,stmaryrd}

\usepackage{bold-extra}

\usepackage{float}
\allowdisplaybreaks
\usepackage{microtype,xspace,bm}
\usepackage{doi}
\usepackage{hyperref}
\usepackage[dvipsnames]{xcolor}
\hypersetup{
	colorlinks=true,
	urlcolor=MidnightBlue,
	linkcolor=MidnightBlue,
	citecolor=MidnightBlue,
	unicode
}
\usepackage[capitalise,nameinlink,noabbrev]{cleveref}
\crefname{equation}{}{}

\usepackage{tikz}
\usepackage{graphicx}
\usetikzlibrary{positioning,arrows.meta,math,shapes.geometric,decorations.pathmorphing}
\newcommand{\spacex}{1.5}
\newcommand{\spacey}{1.125}
\newcommand{\labelspace}{1.5}
\newcommand{\polygonshift}{0.6}
\newlength{\nodeline}
\newlength{\arrowline}
\definecolor{color_gadget_PHP}{RGB}{219, 48, 122}
\colorlet{color_gadget_PHP_inner}{color_gadget_PHP!10!white}
\colorlet{color_gadget_PHP_label}{color_gadget_PHP!70!black}

\definecolor{color_gadget_SOPL}{RGB}{255, 153, 0}
\colorlet{color_gadget_SOPL_inner}{color_gadget_SOPL!10!white}
\colorlet{color_gadget_SOPL_label}{color_gadget_SOPL!70!black}

\definecolor{color_gadget_PATHPHP}{RGB}{0, 153, 255}
\colorlet{color_gadget_PATHPHP_inner}{color_gadget_PATHPHP!10!white}
\colorlet{color_gadget_PATHPHP_label}{color_gadget_PATHPHP!70!black}

\definecolor{color_gadget_ITER}{RGB}{153, 204, 0}
\colorlet{color_gadget_ITER_inner}{color_gadget_ITER!10!white}
\colorlet{color_gadget_ITER_label}{color_gadget_ITER!70!black}

\definecolor{color_gadget_EOPL}{RGB}{0, 153, 255}
\colorlet{color_gadget_EOPL_inner}{color_gadget_EOPL!10!white}
\colorlet{color_gadget_EOPL_label}{color_gadget_EOPL!70!black}

\tikzstyle{node} = [circle, line width=\nodeline, draw = black, fill = white, inner sep = 0mm, minimum size = 3.5mm]
\tikzstyle{node_small} = [node, circle, line width = 0.25mm, minimum size = 2.5mm]
\tikzstyle{solution} = [fill=red!90!,draw=black!50!red]
\tikzstyle{side} = [fill=Goldenrod,draw=Brown]
\tikzstyle{node_text} = [] 

\tikzstyle{node_regular} = [node]
\tikzstyle{node_regular_small} = [node_small]

\tikzstyle{node_solution} = [node, solution]
\tikzstyle{node_solution_small} = [node_small,solution]

\tikzstyle{node_a} = [node, side, rectangle, minimum size = 4.3mm]
\tikzstyle{node_a_solution} = [node_a, solution]
\tikzstyle{node_a_small} = [node_small, side, rectangle, minimum size = 2.15mm]
\tikzstyle{node_a_solution_small} = [node_a_small,solution]

\tikzstyle{node_b} = [node, side, diamond, minimum size = 6.2mm]
\tikzstyle{node_b_small} = [node_a_small, diamond, minimum size = 3.1mm]

\tikzstyle{node_notice} = [node, draw = Green!20!LimeGreen, line width=2.5\nodeline, fill=none, minimum size = 9mm, dotted]
\tikzstyle{node_notice_small} = [node_notice, line width=2*\nodeline, minimum size = 8mm]

\tikzstyle{naive} = [minimum size = 4.5mm,rectangle]
\tikzstyle{node_regular_intro} = [node,fill=Gray!10!white]

\tikzstyle{edge} = [-{Latex[round]}, line width=\arrowline]

\tikzstyle{edge_regular} = [edge]
\tikzstyle{edge_regular_small} = [-{Latex[round]}, line width = 0.25mm, shorten < = 3pt, shorten >=3pt]

\tikzstyle{edge_php} = [edge, color=color_gadget_PHP!70!black]
\tikzstyle{edge_php_small} = [edge_php, edge_regular_small]

\tikzstyle{edge_eopl} = [edge, color=color_gadget_EOPL!70!black]
\tikzstyle{edge_eopl_small} = [edge_eopl, edge_regular_small]

\tikzstyle{edge_iter} = [edge, color=color_gadget_ITER!70!black]
\tikzstyle{edge_iter_small} = [edge_iter, edge_regular_small]

\tikzstyle{edge_pathphp_small} = [line width = 0.25mm, -{Latex[round]}, decorate, decoration={snake, segment length=2.5mm, amplitude=1mm, pre length=7pt,post length=8pt}, shorten < = 3pt, shorten >=3pt, color=color_gadget_PATHPHP_label]

\tikzstyle{gadget} = [rounded corners, line width = 0.4mm, dashed]

\tikzstyle{gadget_PHP} = [gadget, draw = color_gadget_PHP, fill=color_gadget_PHP_inner]
\tikzstyle{gadget_PHP_small} = [gadget_PHP, line width = 0.2mm]

\tikzstyle{gadget_SOPL} = [gadget, draw = color_gadget_SOPL, fill=color_gadget_SOPL_inner]
\tikzstyle{gadget_SOPL_small} = [gadget_SOPL, line width = 0.2mm]

\tikzstyle{gadget_PATHPHP} = [gadget, draw = color_gadget_PATHPHP, fill=color_gadget_PATHPHP_inner]
\tikzstyle{gadget_PATHPHP_small} = [gadget_PATHPHP, line width = 0.2mm]

\tikzstyle{gadget_ITER} = [gadget, draw = color_gadget_ITER, fill=color_gadget_ITER_inner]
\tikzstyle{gadget_ITER_small} = [gadget_ITER, line width = 0.2mm]

\tikzstyle{gadget_EOPL} = [gadget, draw = color_gadget_EOPL, fill=color_gadget_EOPL_inner]
\tikzstyle{gadget_EOPL_small} = [gadget_EOPL, line width = 0.2mm] 

\usepackage[font=small,labelfont=bf]{caption} 
\usepackage{subcaption}

\theoremstyle{definition}
\newtheorem{definition}{Definition}

\theoremstyle{plain}
\newtheorem{theorem}{Theorem}
\newtheorem{lemma}{Lemma}

\theoremstyle{remark}
\newtheorem{remark}{Remark}

\Crefname{claim}{Claim}{Claims}

\newcommand{\newclass}[2]{\newcommand{#1}{{\text{\upshape\sffamily #2}}\xspace}}

\newclass{\NP}{NP}
\newclass{\FP}{FP}
\newclass{\TFNP}{TFNP}
\newclass{\PLS}{PLS}
\newclass{\PPA}{PPA}
\newclass{\PPAD}{PPAD}
\newclass{\PPADS}{PPADS}
\newclass{\PPP}{PPP}
\newclass{\PWPP}{PWPP}
\newclass{\CLS}{CLS}
\newclass{\EOPL}{EOPL}
\newclass{\SOPL}{SOPL}
\newclass{\UEOPL}{UEOPL}
\newclass{\cA}{A}
\newclass{\cB}{B}

\newcommand{\newprob}[2]{\newcommand{#1}{{\text{\upshape\scshape #2}}\xspace}}
\newprob{\eol}{EoL}
\newprob{\eolLong}{End-of-Line}
\newprob{\sol}{SoL}
\newprob{\solLong}{Sink-of-Line}
\newprob{\iter}{Iter}
\newprob{\sod}{SoD}
\newprob{\sodLong}{Sink-of-DAG}
\newprob{\kkt}{KKT}
\newprob{\pA}{A}
\newprob{\pB}{B}
\newprob{\grid}{Grid}
\newprob{\eopl}{EoPL}
\newprob{\eoplLong}{End-of-Potential-Line}
\newprob{\ueoplLong}{Unique-EoPL}
\newprob{\eoml}{EoML}
\newprob{\eomlLong}{End-of-Metered-Line}
\newprob{\sopl}{SoPL}
\newprob{\soplLong}{Sink-of-Potential-Line}
\newprob{\pigeoncircuit}{Pigeon-Circuit}
\newprob{\php}{PHP}
\newprob{\bphp}{bij-PHP}
\newprob{\iphp}{inj-PHP}
\newprob{\iphpLong}{Injective-Pigeonhole}
\newprob{\pathbphp}{Path-bij-PHP}
\newprob{\pathiphp}{Path-inj-PHP}

\newcommand{\poly}{\textup{poly}}
\newcommand{\nul}{\textup{\textsf{null}}}

\begin{document}

\mbox{}\vspace{8mm}

\begin{center}
{\huge Further Collapses in $\TFNP$}
\\[1.3cm] \large
	
\setlength\tabcolsep{1.2em}
\begin{tabular}{cccc}
Mika G\"o\"os&
Alexandros Hollender&
Siddhartha Jain&
Gilbert Maystre\\[-1mm]
\small\slshape EPFL &
\small\slshape University of Oxford &
\small\slshape EPFL &
\small\slshape EPFL
\end{tabular}

\vspace{1mm}
\begin{tabular}{ccc}
William Pires&
Robert Robere&
Ran Tao\\[-1mm]
\small\slshape McGill University &
\small\slshape McGill University &
\small\slshape McGill University
\end{tabular}

\vspace{6mm}
	
\large

	
\vspace{4mm}
\end{center}

\begin{quote}
\noindent\small
{\bf Abstract.}~
We show $\EOPL=\PLS\cap\PPAD$. Here the class $\EOPL$ consists of all total search problems that reduce to the {\scshape End-of-Potential-Line} problem, which was introduced in the works by Hub{\'a}{\v{c}}ek and Yogev ({\footnotesize SICOMP 2020}) and Fearnley et~al.~({\footnotesize JCSS 2020}). In particular, our result yields a new simpler proof of the breakthrough collapse $\CLS=\PLS\cap\PPAD$ by Fearnley et~al.~({\footnotesize STOC 2021}). We also prove a companion result $\SOPL=\PLS\cap\PPADS$, where $\SOPL$ is the class associated with the {\scshape Sink-of-Potential-Line} problem.
\end{quote}

\section{Introduction}

Our main results are two collapses of total $\NP$ search problem ($\TFNP$) classes.
\begin{theorem}\label{thm:EOPL-collapse}
$\EOPL = \PLS\cap\PPAD$.
\end{theorem}
\begin{theorem}\label{thm:SOPL-collapse}
$\SOPL = \PLS\cap\PPADS$.
\end{theorem}

In particular, \cref{thm:EOPL-collapse} answers a question asked by Daskalakis in his Nevanlinna Prize lecture~\cite[Open Question 10]{Daskalakis2019}. Let us explain what these collapses mean and how they fit into the diverse complexity zoo of search problem classes, as summarised in \cref{fig:classes}. The classes $\PLS$, $\PPAD$, $\PPADS$ are classical. They were all introduced in the original pioneering works~\cite{Megiddo1991, Johnson1988,Papadimitriou1994} that founded the theory of $\TFNP$. To define these classes, it is most convenient to describe a canonical complete problem for each class. (See \cref{sec:grid} for more formal definitions).
\begin{description}
\item[\bf $\PLS$:] $\sodLong$ ($\sod$). We are given \emph{implicit access} to a directed graph $G=(V,E)$ that is acyclic, has out-degree at most $1$, and has exponentially many nodes, $|V|=2^n$. The graph is described by a $\poly(n)$-sized circuit: for any node $v\in V$, we can compute its unique \emph{successor} (out-neighbour) $u$, if any, and also an integer \emph{potential}, which is guaranteed to increase along the direction of the edge~$(v,u)$. The goal is to find a \emph{sink} node (in-degree~$\geq 1$, out-degree~0).

\item[\bf $\PPAD$:] $\eolLong$ ($\eol$). We are given access to a directed graph $G=(V,E)$ that has in/out-degree at most $1$, and has $|V|=2^n$ nodes. The graph is described by a $\poly(n)$-sized circuit: for any $v\in V$, we can compute its \emph{successor} $u$ and \emph{predecessor}~$u'$, if any. We are guaranteed that if $v$'s successor is $u$, then $u$'s predecessor is $v$, and vice versa. In addition, we are given the name of a \emph{distinguished source} node $v^*$ (in-degree~0, out-degree~1). The goal is to find any source or sink other than $v^*$.

\item[\bf $\PPADS$:] $\solLong$ ($\sol$). Same as $\eol$ except the goal is to find a sink.
\end{description}

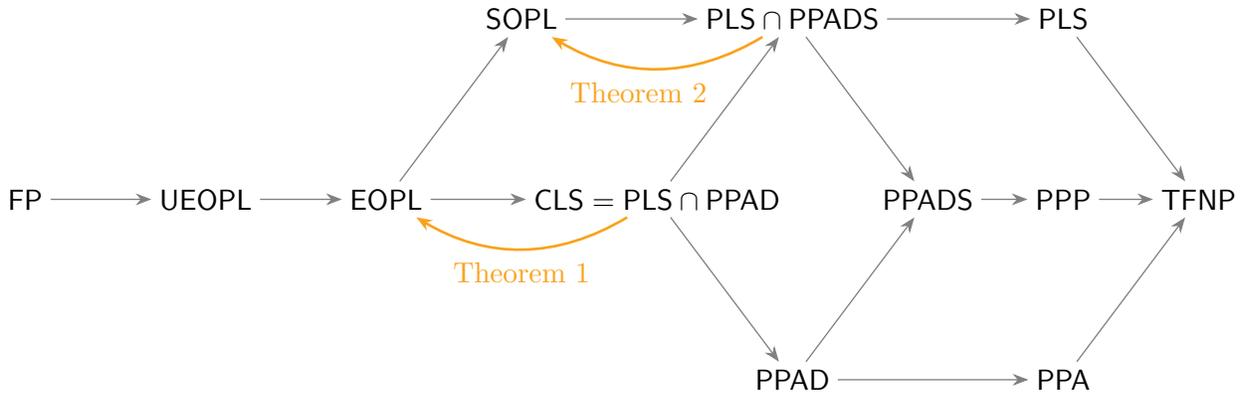
\begin{figure}[t]
\centering
\begin{tikzpicture}[scale=0.8]
\tikzset{inner sep=0,outer sep=3}


\begin{scope}[xscale=1.5]
\node (FP) at (-5.5,0) {$\FP$};
\node (UEOPL) at (-3.5,0) {$\UEOPL$};
\node (EOPL) at (-1.5,0) {$\EOPL$};
\node (SOPL) at (0,3) {$\SOPL$};
\node (PLS-PPAD) at (1.5,0) {$\CLS=\PLS\cap\PPAD$};
\node (PLS-PPADS) at (3,3) {$\PLS\cap\PPADS$};
\node (PPAD) at (3,-3) {$\PPAD$};
\node (PPADS) at (4.5,0) {$\PPADS$};
\node (PLS) at (6,3) {$\PLS$};
\node (PPP) at (6,0) {$\PPP$};
\node (PPA) at (6,-3) {$\PPA$};
\node (TFNP) at (7.5,0) {$\TFNP$};
\end{scope}

\path[-{Stealth[length=6pt]},line width=.5pt,gray]
(FP) edge (UEOPL)
(UEOPL) edge (EOPL)
(EOPL) edge (SOPL)
(EOPL) edge (PLS-PPAD)
(PLS-PPAD) edge (PLS-PPADS)
(PLS-PPAD) edge (PPAD)
(PLS-PPADS) edge (PPADS)
(PLS-PPADS) edge (PLS)
(SOPL) edge (PLS-PPADS)
(PPAD) edge (PPADS)
(PPAD) edge (PPA)
(PPADS) edge (PPP)
(PLS) edge (TFNP)
(PPP) edge (TFNP)
(PPA) edge (TFNP);

\hypersetup{hidelinks}
\tikzset{new/.style={-{Stealth[length=6pt]},line width=1pt,YellowOrange}}

\draw[new]
(PLS-PPAD) edge[bend left=31]
node[midway,below,inner sep=2pt] {\cref{thm:EOPL-collapse}}
(EOPL);
\draw[new]
(PLS-PPADS) edge[bend left=31]
node[midway,below,xshift=-7,inner sep=2pt] {\cref{thm:SOPL-collapse}}
(SOPL);

\end{tikzpicture}
\vspace{1mm}
\caption{Class diagram for $\TFNP$ with new inclusions highlighted. An arrow $\cA\to\cB$ denotes $\cA\subseteq \cB$.}
\label{fig:classes}
\vspace{1mm}
\end{figure}

\paragraph{Modern classes.}
Research in the past decade has studied several relatively weak classes of search problems that lie below~$\PLS$ and $\PPAD$. The intersection class $\PLS\cap\PPAD$ is, of course, one immediate such example. This class, however, feels quite artificial at first glance. It does not seem to admit any ``natural'' complete problem. Motivated by this, Daskalakis and Papadimitriou~\cite{Daskalakis2011} introduced the \emph{continuous local search} class $\CLS\subseteq \PLS\cap\PPAD$, which, by its very definition, admits natural complete problems related to the local optimisation of continuous functions over the real numbers (computed by arithmetic circuits). The class $\CLS$ is exceptional in that it captures the complexity of \emph{real continuous} optimisation problems, while most classical search problem classes are designed to capture \emph{combinatorial principles}, often phrased in terms of directed graphs.

In order to understand $\CLS$ from a more combinatorial perspective, Hub{\'a}{\v{c}}ek and Yogev~\cite{Hubacek2020} and Fearnley, Gordon, Mehta, and Savani~\cite{Fearnley2020} introduced the class~$\EOPL\subseteq\CLS$, whose complete problem is the namesake $\eoplLong$ ($\eopl$) problem, defined below. (The paper~\cite{Hubacek2020} initially defined a more restricted ``metered'' version of this problem, but we use the formulation from~\cite{Fearnley2020}, which they prove is equivalent to the one from~\cite{Hubacek2020}.) It is also natural to define a \emph{sink-only} version of $\EOPL$ as suggested by~\cite{Goos2018}.
\begin{description}
\item[\bf $\EOPL$:] $\eoplLong$ ($\eopl$). We are given access to a directed graph $G=(V,E)$ that is acyclic, has in/out-degree at most $1$, and has $|V|=2^n$ nodes; that is, $G$ is a disjoint union of directed paths. The graph is described by a $\poly(n)$-sized circuit: for any node we can compute its successor and predecessor, if any, and also an integer potential, which is guaranteed to increase along the directed edges. In addition, we are given the name of a distinguished source $v^*$. The goal is to find any source or sink other than $v^*$.

\item[\bf $\SOPL$:] $\soplLong$ ($\sopl$). Same as $\eopl$ except the goal is to find a sink.
\end{description}
It is comforting to know that the definition of $\EOPL$ is robust: Ishizuka~\cite{Ishizuka2021} showed that a version of $\eopl$ that guarantees $\poly(n)$ many distinguished sources is still equivalent (via polynomial-time reductions) to the above standard version with a single source.

Fearnley et al.~\cite{Fearnley2020} also defined a more restricted subclass $\UEOPL\subseteq\EOPL$ where the complete problem is $\ueoplLong$, a version of $\eopl$ with a \emph{unique} directed path. They showed that this class contains many important search problems with unique witnesses, such as unique sink orientations, linear complementary problems, {\scshape Arrival}~\cite{Dohrau2017,Gaertner2018}. Other problems known to lie in \UEOPL are a restricted version of the Ham-Sandwich problem~\cite{Chiu2020} and a pizza cutting problem~\cite{Schnider2021}. Fearnley et al.~\cite{Fearnley2020} conjecture that $\UEOPL\neq\EOPL$.

\paragraph{A surprising collapse.}
In a breakthrough, Fearnley, Goldberg, Hollender, and Savani~\cite{Fearnley2021} showed that, despite appearances to the contrary, $\CLS=\PLS\cap\PPAD$.  This goes against the conjecture of Daskalakis and Papadimitriou~\cite{Daskalakis2011} that the classes are distinct, a belief which underlied much of their original motivation for introducing $\CLS$. The nontrivial direction of the collapse is a reduction from a canonical complete problem $\sod\curlywedge\eol\in \PLS\cap\PPAD$ (defined below) to a problem $\kkt\in\CLS$, which involves computing a Karush--Kuhn--Tucker point of a smooth function. We may summarise the main technical result of Fearnley et al.~\cite{Fearnley2021} as
\begin{align} \label{eq:kkt}
\sod\curlywedge\eol &~\leq~ \kkt
&&\text{which implies}\enspace \PLS\cap\PPAD\subseteq\CLS.
\end{align}
Here we use $\leq$ to denote a polynomial-time reduction between search problems. The operator~$\curlywedge$ produces the \emph{meet} of two search problems: the input to problem $\pA\curlywedge\pB$ is a pair $(x,y)$ where $x$ is an instance of $\pA$ and $y$ is an instance of $\pB$ and the goal is to output either a solution to~$x$ or to $y$. Then $\sod\curlywedge\eol$ is the canonical (albeit ``unnatural'') complete problem for $\PLS\cap\PPAD$~\cite{Daskalakis2011}.

\paragraph{Our new collapses.}
Our main results, \cref{thm:EOPL-collapse,thm:SOPL-collapse}, follow from two new reductions, the first one of which strengthens the reduction \cref{eq:kkt} from~\cite{Fearnley2021}:
\begin{align} 
\sod\curlywedge\eol &~\leq~ \eopl
&&\text{which implies}\enspace \PLS\cap\PPAD\subseteq\EOPL,
\label{eq:eopl} \\
\sod\curlywedge\sol &~\leq~ \sopl
&&\text{which implies}\enspace \PLS\cap\PPADS\subseteq\SOPL.
\end{align}
These reductions are between purely combinatorially defined search problems. In the case of~\cref{eq:eopl}, this bypasses the continuous middle-man of $\CLS$ and makes our reduction relatively simple to describe. In particular, we get a new simpler proof of the breakthrough collapse of~\cite{Fearnley2021} by combining~\cref{eq:eopl} with the inclusion $\EOPL\subseteq\CLS$ proved by~\cite{Hubacek2020}. Furthermore, the new collapse implies that problems related to Tarski's fixpoint theorem~\cite{Etessami2020} and to a colourful version of Carath{\'e}odory's theorem~\cite{Meunier2017} lie in \EOPL.

\paragraph{A further surprise?}
Given that the collapse $\CLS=\PLS\cap\PPAD$ was considered extremely surprising by most experts, how surprised should we be by the further collapse
\[
\EOPL\, =\, \CLS\, =\, \PLS\cap\PPAD\enspace ?
\]
Fearnley et al.~\cite{Fearnley2020} wrote regarding $\EOPL$ vs.\ $\CLS$ that ``we actually think it could go either way.'' In the wake of their breakthrough, the paper~\cite{Fearnley2021} explicitly conjectured $\EOPL\neq\CLS$.

For the authors of the present paper, the new collapse did come as an utter shock. When we began work on this project, our intuitions convinced us that, again, $\EOPL\neq \CLS$, a conjecture which had just found its way to the second author's PhD thesis~\cite[Section~7.5]{Hollender2021}. In our convictions, we set out to prove this separation in the \emph{black-box model} where, instead of circuits, the directed graphs are described by black-box oracles. We tried in vain for nine months. The upshot is that \cref{thm:EOPL-collapse,thm:SOPL-collapse} now crush this possibility, as they hold even in the black-box model.

\section{A Unified View: The \texorpdfstring{$\grid$}{Grid} Problem}\label{sec:grid}

In this section we formally define all the problems of interest. We take the unusual approach of defining a single problem (which we call the $\grid$ problem) with various parameters which can be tweaked to obtain all of the problems we study in this paper. This mainly serves two purposes. First of all, it is particularly convenient for presenting our reductions, since it allows us to combine instances from different problems more easily. The second reason is that we believe that this unified view of seemingly very different problems is of independent interest.

\paragraph{The $\grid$ problem.}

For $n \in \mathbb{N}$, let $[n] \coloneqq \{1,2,\dots,n\}$. We define a general problem on a grid $[N] \times [M]$, where $N$ and $M$ should be thought of as being (potentially) exponentially large. The problem involves $A$ paths starting from column $1$ ($[N] \times \{1\}$) and moving from column $i$ ($[N] \times \{i\}$) to column $i+1$ ($[N] \times \{i+1\}$). On the last column ($[N] \times \{M\}$) there are at most $B$ valid ends of paths. If paths are not allowed to merge, then by the Pigeonhole Principle $A > B$ ensures the existence of a solution, i.e., a path that does not end at a valid position on the last column. If paths are allowed to merge, then a solution is guaranteed to exist as long as $B = 0$. To make things more precise, the paths start from nodes $1$ to $A$ in the first column (i.e., $[A] \times \{1\}$), and the valid termination points are nodes $1$ to $B$ in the last column (i.e., $[B] \times \{M\}$).

In more detail, we are given a boolean circuit $S\colon [N] \times [M] \to [N] \cup \{\nul\}$, the \emph{successor circuit}, which allows us to efficiently compute the outgoing edge at a node. If $S(x,y) = \nul$, then $(x,y)$ does not have an outgoing edge. Otherwise, there is an outgoing edge from $(x,y)$ to $(S(x,y),y+1)$. The problem also has two parameters which are used to tweak the definition: $r$ (\emph{reversible}) and $b$ (\emph{bijective}). Intuitively, when $r=1$, we change the representation of paths to make them \emph{reversible}. Namely, in addition to the successor circuit $S$, we are also given access to a \emph{predecessor circuit} $P \colon [N] \times [M] \to [N] \cup \{\nul\}$, which, analogously to $S$, allows us to efficiently compute the incoming edge at a node. In particular, when $r=1$, every node can have at most one incoming edge, i.e., two paths cannot merge. When $r=1$, the other parameter $b$ is used to introduce additional solutions. Namely, when $b=1$, then we do not allow any new paths apart from the original $A$ paths, and we also require that all $B$ valid ends of paths are actually reached by a path. The combination $r=0, b=1$ is not allowed.

We use the term \emph{sink} to refer to a node with at least one incoming edge but no outgoing edge. Similarly, a \emph{source} is a node with an outgoing edge but no incoming edge. The formal definition of the problem is as follows.

\begin{definition}\label{definition:grid_problem}
In the \grid problem, given $N,M,A,B$ with $N \geq A > B \geq 0$ and $M \geq 2$, boolean circuits $S, P \colon [N] \times [M] \to [N] \cup \{\nul\}$, and bits $r,b \in \{0,1\}$, output any of the following:
\begin{enumerate}
    \item $x \in [A]$ such that $S(x,1) = \nul$, \hfill \emph{(missing pigeon/source)}
    \item $x \in [N]$ such that $S(x,M-1) > B$, \hfill \emph{(invalid hole/sink)}
    \item $x \in [N]$ and $y \in [M-2]$ such that\\
    $S(x,y) \neq \nul$ and $S(S(x,y),y+1) = \nul$, \hfill \emph{(pigeon interception/sink)}
    \item If $r=1$ and $b=1$:
    \begin{enumerate}
        \item $(x,y) \in ([N] \times [M-1]) \setminus ([A] \times \{1\})$ such that\\
        $S(x,y) \neq \nul$ and $P(x,y) = \nul$, or \hfill \emph{(pigeon genesis/source)}
        \item $x \in [B]$ such that $P(x,M) = \nul$. \hfill \emph{(empty hole/sink)}
    \end{enumerate}
\end{enumerate}
We also enforce the following two conditions syntactically:
\begin{itemize}
    \item If $r=0$, then $b=0$ and $B=0$.
    \item If $r=1$, then the successor and predecessor circuits are consistent, which can be enforced as follows. The circuit $S$ is replaced by the circuit $\overline{S}$, which on input $(x,y)$ computes $x ' \coloneqq S(x,y)$ and outputs $x'$, unless $(x',y+1) \notin [N] \times [M]$ or $P(x',y+1) \neq x$, in which case it outputs $\nul$. Similarly, the circuit $P$ is replaced by the circuit $\overline{P}$, which on input $(x,y)$ computes $x ' \coloneqq P(x,y)$ and outputs $x'$, unless $(x',y-1) \notin [N] \times [M]$ or $S(x',y-1) \neq x$, in which case it outputs $\nul$.
\end{itemize}
\end{definition}

\paragraph{Canonical complete problems as special cases of $\grid$.}

As defined above, the inputs $N, M, A, B$ of the $\grid$ problem are completely unrestricted, apart from the natural restrictions $N \geq A > B \geq 0$ and $M \geq 2$. By imposing various additional restrictions on these inputs, we obtain the following canonical complete problems; see~\cref{figure:grid_problem_summary}. (Here $\iphp$/$\bphp$ stand for Injective/Bijective Pigeonhole Principle.)
\begin{itemize}
    \item $\sod$: $r=0, b=0, A=1, B=0$. (\PLS-complete)
    \item $\sopl$: $r=1, b=0, A=1, B=0$. (\SOPL-complete)
    \item $\eopl$: $r=1, b=1, A=1, B=0$. (\EOPL-complete)
    \item $\iphp$: $r=1, b=0, M=2, N=A=B+1$. (\PPADS-complete)
    \item $\bphp$: $r=1, b=1, M=2, N=A=B+1$. (\PPAD-complete)
\end{itemize}
Note that beyond those restrictions, the inputs are left unrestricted. For example, in $\sod$, the input $M$ can be very large, which is indeed needed for the problem to be \PLS-complete.

\begin{figure}[t]
\centering
\begin{subfigure}[b]{0.33\textwidth}
    \centering
    \begin{tikzpicture}[y=-1cm, scale=1]
\coordinate (p11) at (\labelspace, \labelspace); \coordinate (p12) at (\labelspace + \spacex, \labelspace); \coordinate (p13) at (\labelspace + 2*\spacex, \labelspace); \coordinate (p14) at (\labelspace + 3*\spacex, \labelspace); 
\coordinate (p21) at (\labelspace, \labelspace + \spacey); \coordinate (p22) at (\labelspace + \spacex, \labelspace + \spacey);	 \coordinate (p23) at (\labelspace + 2*\spacex, \labelspace + \spacey); \coordinate (p24) at (\labelspace + 3*\spacex, \labelspace + \spacey); 
\coordinate (p31) at (\labelspace, \labelspace + 2*\spacey); \coordinate (p32) at (\labelspace + \spacex, \labelspace + 2*\spacey); \coordinate (p33) at (\labelspace + 2*\spacex, \labelspace + 2*\spacey); \coordinate (p34) at (\labelspace + 3*\spacex, \labelspace + 2*\spacey); 
\coordinate (p41) at (\labelspace, \labelspace + 3*\spacey); \coordinate (p42) at (\labelspace + \spacex, \labelspace + 3*\spacey); \coordinate (p43) at (\labelspace + 2*\spacex, \labelspace + 3*\spacey); \coordinate (p44) at (\labelspace + 3*\spacex, \labelspace + 3*\spacey); 

\tikzstyle{node_regular} = [node_regular_intro]

\node[node_a]        (P11) at (p11) {};
\node[node_regular]  (P12) at (p12) {};
\node[node_regular]  (P13) at (p13) {};
\node[node_solution] (P14) at (p14) {};
\node[node_regular]  (P21) at (p21) {};
\node[node_regular]  (P22) at (p22) {};
\node[node_regular]  (P23) at (p23) {};
\node[node_regular]  (P24) at (p24) {};
\node[node_regular]  (P31) at (p31) {};
\node[node_regular]  (P32) at (p32) {};
\node[node_solution] (P33) at (p33) {};
\node[node_regular]  (P34) at (p34) {};
\node[node_regular]  (P41) at (p41) {};
\node[node_regular]  (P42) at (p42) {};
\node[node_regular]  (P43) at (p43) {};
\node[node_regular]  (P44) at (p44) {};

\draw[edge_regular] (P11) -- (P12);
\draw[edge_regular] (P12) -- (P23);
\draw[edge_regular] (P22) -- (P23);
\draw[edge_regular] (P23) -- (P14);
\draw[edge_regular] (P21) -- (P32);
\draw[edge_regular] (P32) -- (P33);
\draw[edge_regular] (P42) -- (P33);

\end{tikzpicture}
    \vspace{2mm}
    \caption{$\sodLong$ ($\sod$)}
    \label{figure:intro_iter}
\end{subfigure}%
\begin{subfigure}[b]{0.33\textwidth}
    \centering
    \begin{tikzpicture}[y=-1cm, scale=1]
\coordinate (p11) at (0, 0); 				 \coordinate (p12) at (\spacex, 0);
\coordinate (p21) at (0, \spacey);	 \coordinate (p22) at (\spacex, \spacey);
\coordinate (p31) at (0, 2*\spacey); \coordinate (p32) at (\spacex, 2*\spacey);
\coordinate (p41) at (0, 3*\spacey); \coordinate (p42) at (\spacex, 3*\spacey);

\tikzstyle{node_regular} = [node_regular_intro]

\node[xshift=-15] at (p11) {\mbox{}};
\node[xshift=15] at (p12) {\mbox{}};

\node[node_a] (P11) at (p11) {};
\node[node_a] (P21) at (p21) {};
\node[node_a_solution] (P31) at (p31) {};
\node[node_a_solution] (P41) at (p41) {};
\node[node_b] (P12) at (p12) {};
\node[node_b] (P22) at (p22) {};
\node[node_b] (P32) at (p32) {};
\node[node_regular] (P42) at (p42) {};

\draw[edge_regular] (P11) -- (P22);
\draw[edge_regular] (P21) -- (P12);

\node[node_notice] at (p32) {};
\end{tikzpicture}
    \vspace{2mm}
    \caption{$\iphp$}
    \label{figure:intro_php}
\end{subfigure}%
\begin{subfigure}[b]{0.33\textwidth}
    \centering
    \begin{tikzpicture}[y=-1cm, scale=1]
\coordinate (p11) at (0, 0); 				 \coordinate (p12) at (\spacex, 0); 				 \coordinate (p13) at (2*\spacex, 0); 				\coordinate (p14) at (3*\spacex, 0);
\coordinate (p21) at (0, \spacey);	 \coordinate (p22) at (\spacex, \spacey); 	 \coordinate (p23) at (2*\spacex, \spacey); 	\coordinate (p24) at (3*\spacex, \spacey);
\coordinate (p31) at (0, 2*\spacey); \coordinate (p32) at (\spacex, 2*\spacey); \coordinate (p33) at (2*\spacex, 2*\spacey); \coordinate (p34) at (3*\spacex, 2*\spacey);
\coordinate (p41) at (0, 3*\spacey); \coordinate (p42) at (\spacex, 3*\spacey); \coordinate (p43) at (2*\spacex, 3*\spacey); \coordinate (p44) at (3*\spacex, 3*\spacey);

\tikzstyle{node_regular} = [node_regular_intro]

\node[node_a]       (P11) at (p11) {};
\node[node_regular] (P12) at (p12) {};
\node[node_regular] (P13) at (p13) {};
\node[node_solution] (P14) at (p14) {};
\node[node_regular] (P21) at (p21) {};
\node[node_regular] (P22) at (p22) {};
\node[node_regular] (P23) at (p23) {};
\node[node_regular] (P24) at (p24) {};
\node[node_regular] (P31) at (p31) {};
\node[node_regular] (P32) at (p32) {};
\node[node_regular] (P33) at (p33) {};
\node[node_solution] (P34) at (p34) {};
\node[node_regular] (P41) at (p41) {};
\node[node_regular] (P42) at (p42) {};
\node[node_solution] (P43) at (p43) {};
\node[node_regular] (P44) at (p44) {};

\draw[edge_regular] (P11) -- (P12);
\draw[edge_regular] (P12) -- (P23);
\draw[edge_regular] (P23) -- (P14);

\draw[edge_regular] (P32) -- (P33);
\draw[edge_regular] (P33) -- (P34);

\draw[edge_regular] (P31) -- (P42);
\draw[edge_regular] (P42) -- (P43);

\node[node_notice] at (p31) {};
\node[node_notice] at (p32) {};
\end{tikzpicture}
    \vspace{2mm}
    \caption{$\soplLong$ ($\sopl$)}
    \label{figure:intro_sopl}
\end{subfigure}
\caption{Examples of $\grid$ problems. Square nodes are valid starts of paths (top-most $A$ nodes in the first column) and diamonds are valid ends of paths (top-most $B$ nodes in the last column). Solutions are drawn in red. However, for visual clarity we highlight the actual sinks rather than the \emph{sink predecessors} as in \cref{definition:grid_problem}. Nodes with a $\nul$ successor are drawn without an outgoing pointer.
(\ref{figure:intro_iter}) has parameters $(r = 0, b = 0, A=1, B=0)$ and defines an $\sod$ instance. Only the successor circuit is drawn, as the predecessor circuit is not used by $\sod$. In particular, directed paths can \emph{merge}, such as for node $(2, 3)$.
(\ref{figure:intro_php}) has parameters $(r=1, b=0, A=N, B=N-1)$ and defines an $\iphp$ instance. The diamond with a green circle would be a solution of $\bphp$ (with $b=1$) but is not a solution of $\iphp$.
(\ref{figure:intro_sopl}) has parameters $(r =1, b= 0, A=1, B=0)$ and defines an $\sopl$ instance. Sources with green circles would be solutions of an $\eopl$ instance (with $b=1$).
}
\label{figure:grid_problem_summary}
\end{figure}
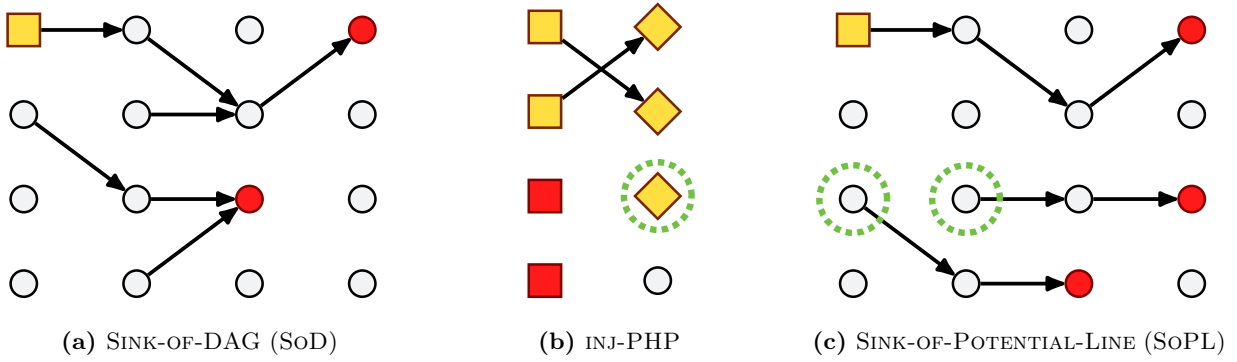

\begin{remark}
Here we have slightly abused notation by calling these problems $\sod$, $\sopl$ and $\eopl$ even though their original definitions (in~\cite{Johnson1988,Goos2018,Fearnley2020}, respectively) do not use a grid structure, and instead come with an additional circuit computing the potential of any node. It is not too hard to see that these grid-versions of the problems are indeed polynomial-time equivalent to the original versions. The main idea is that the grid implicitly provides a potential value for every node $(x,y)$, namely its column number $y$. Thus, given such a problem on a grid, it is easy to define a potential circuit by simply assigning the potential value $y$ to any node $(x,y)$ of the grid.

The other direction is slightly more involved. Consider an instance of one of the original problems with vertex set $V = [N]$ and potential values lying in $P = [M]$. Without loss of generality, we can assume that along any edge the potential increases by exactly one. Indeed, this was proved explicitly by~\cite{Fearnley2020} when they reduced $\eopl$ to $\eoml$, and the same idea applies to $\sod$ and $\sopl$ as well. The reduction to the grid-version of the problem is then obtained by identifying a vertex $x \in V$ that has potential $p \in P$ with the node $(x,p)$ on the $[N] \times [M]$ grid.
\end{remark}

The following is essentially folklore (see, e.g., \cite{Buresh2004}), so we only provide a brief proof sketch.

\begin{lemma}\label{lem:PPAD-as-PHP}
$\iphp$ and $\bphp$ are respectively \PPADS- and \PPAD-complete.
\end{lemma}

\begin{proof}[Proof Sketch]
To see that $\bphp$ lies in \PPAD, we can reduce to $\eol$ (see, e.g., \cite{Daskalakis2009} for a formal definition) with vertex set $V = [N] \times [2]$ as follows: add a directed edge from node $(x,2)$ to node $(x,1)$ for all $x \leq A-1$. By taking $(x,A)$ as the distinguished source node, this yields an $\eol$ instance with the same solutions as the original $\bphp$ instance. On the other hand, given an instance of $\eol$ with vertex set $V = [N]$ and distinguished source node $N$ (without loss of generality), we construct an instance of $\bphp$ on $[N] \times [2]$ as follows: for any isolated vertex $x \in [N]$, create an edge from $(x,1)$ to $(x,2)$; for any edge from $x$ to $y$, create an edge from $(x,1)$ to $(y,2)$. This simple reduction proves the \PPAD-hardness of $\bphp$. The exact same constructions can be used to prove that $\iphp$ is \PPADS-complete, by reducing to and from the $\sol$ problem (formally defined by Beame et al.~\cite{Beame1998}, who call it {\scshape Sink}).
\end{proof}

In \cref{sec:discussion}, we briefly explain how an extended version of the $\grid$ problem can be used to also capture \PPP, the class defined by Papadimitriou~\cite{Papadimitriou1994} to capture a version of the Pigeonhole Principle where edges can only be computed efficiently in the forward direction. We do not currently see any natural way of extending the definition of the $\grid$ problem so that it also captures the class \PPA.

\section{Path-Pigeonhole Problems}

In this section we use the $\grid$ problem to define some interesting extensions of the two pigeonhole problems. Namely, we consider the case where, instead of just two columns, there are many columns. In a certain sense, this corresponds to allowing the pigeons to travel for a long time before reaching a hole. In particular, we can no longer efficiently tell in which hole a given pigeon will land. This allows us to show that the problems remain hard even when there are significantly more pigeons than holes. This fact, stated in \cref{lem:inj-path-php} below, will be crucial to obtain our main result later.

Let $f\colon \mathbb{N} \to \mathbb{N}$ be a polynomial-time computable function with $f(t) > t$. In this section, we consider the following restrictions of $\grid$:

\begin{itemize}
    \item $\pathiphp_f$: $r=1, b=0, A = f(B)$.
    \item $\pathbphp_f$: $r=1, b=1, A = f(B)$.
\end{itemize}

The following lemma is an important ingredient for the proof of our main result.

\begin{lemma}\label{lem:inj-path-php}
Let $f(t)>t$ be polynomial-time computable.
There exists a reduction $\iphp\leq \pathiphp_f$ that maps an instance with parameters $(A,B) = (T+1,T)$ to an instance with parameters $(A,B) = (f(T),T)$ and $(N,M) = (f(T),f(T)-T+1)$.
\end{lemma}

\begin{proof}
The idea behind this reduction is very simple. Intuitively, we have the ability to ``merge'' $T+1$ paths into $T$ paths by using the $\iphp$ instance. Namely, we can go from having $T+1$ paths on some column $i$ to having only $T$ paths on the next column $i+1$, and such that finding a ``mistake'', i.e., a path that stops between the two columns, requires solving the $\iphp$ instance. In particular, if we start with some $N$ paths, where $N \geq T+1$, then we can ``merge'' those paths into $N-1$ paths by ``merging'' the first $T+1$ paths into $T$ paths, and leaving the remaining $N-(T+1)$ paths unchanged. Applying this idea repeatedly, we can ``merge'' $f(T)$ paths into just $T$ paths in $f(T)-T$ steps. This results in an instance of $\pathiphp_f$ with $f(T)-T+1$ columns, where every solution yields a solution to the $\iphp$ instance; see \cref{figure:reduction_php_pathphp}. More formally, let $(S,P)$ denote an instance of $\iphp$ with parameters $A=T+1$ and $B=T$. Without loss of generality, we can assume that no pigeon goes to the invalid hole, i.e., $S(x,1) \neq T+1$ for all $x \in [T+1]$. Indeed, if there is such an edge, we can just remove it and this does not change the set of solutions; the node pointing to the invalid hole was a solution before, and now it is still a solution, because it has no successor. We construct an instance $(\widehat{S}, \widehat{P})$ of $\pathiphp_f$ on the $[N] \times [M]$ grid, where $N=f(T)$, $M=f(T)-T+1$, $A=f(T)$ and $B=T$. The successor circuit $\widehat{S}$ is defined as follows:
\begin{equation*}
    \widehat{S}(x,y)
    ~\coloneqq~
    \begin{cases}
    S(x,1) & \text{if } x \in [T+1] \text{ and } y \in [M-1],\\
    x-1 & \text{if } T+2 \leq x \leq f(T) - y + 1 \text{ and } y \in [M-1],\\
    \nul & \text{otherwise},
    \end{cases}
\end{equation*}
and the predecessor circuit $\widehat{P}$ is then defined accordingly to be consistent with $\widehat{S}$. Both circuits can be constructed in polynomial time, given $S$ and $P$, and given that $f$ can be computed in polynomial time. It is straightforward to check that any solution of the constructed instance yields a solution to the original $\iphp$ instance.
\end{proof}

The same proof idea also yields that $\bphp\leq \pathbphp_f$.

\begin{figure}[t!]
\centering
\begin{tikzpicture}[y=-1cm, scale=1]
\newcommand{\labelshift}{-.9}
\newcommand{\dx}{-3}

\coordinate (p11) at (0, 0); \coordinate (p12) at (\spacex, 0); \coordinate(p13) at (2*\spacex, 0); \coordinate (p14) at (3*\spacex, 0); \coordinate (p15) at (4*\spacex, 0);
\coordinate (p21) at (0, \spacey);
\coordinate (p31) at (0, 2*\spacey);
\coordinate (p41) at (0, 3*\spacey);
\coordinate (p51) at (0, 4*\spacey);
\coordinate (p61) at (0, 5*\spacey);

\coordinate (p22) at (p12 |- p21); \coordinate(p23) at (p13 |- p21); \coordinate(p24) at (p14 |- p21); \coordinate(p25) at (p15 |- p21);
\coordinate (p32) at (p12 |- p31); \coordinate(p33) at (p13 |- p31); \coordinate(p34) at (p14 |- p31); \coordinate(p35) at (p15 |- p31);
\coordinate (p42) at (p12 |- p41); \coordinate(p43) at (p13 |- p41); \coordinate(p44) at (p14 |- p41); \coordinate(p45) at (p15 |- p41);
\coordinate (p52) at (p12 |- p51); \coordinate(p53) at (p13 |- p51); \coordinate(p54) at (p14 |- p51); \coordinate(p55) at (p15 |- p51);
\coordinate (p62) at (p12 |- p61); \coordinate(p63) at (p13 |- p61); \coordinate(p64) at (p14 |- p61); \coordinate(p65) at (p15 |- p61);

\coordinate (c11) at ([shift=({-\polygonshift,-\polygonshift})]p11);
\coordinate (c14) at ([shift=({\polygonshift,-\polygonshift})]p14);
\coordinate (c54) at ([shift=({\polygonshift,\polygonshift})]p54);
\coordinate (c51) at ([shift=({-\polygonshift, \polygonshift})]p51);

\node[node_text] at (1.5*\spacex, \labelshift) {\textcolor{color_gadget_PATHPHP_label}{$\pathiphp_f$}};
\draw[gadget_PATHPHP] (c11) -- (c14) -- (c54) -- (c51) -- cycle;

\node[node_a] (P11) at (p11) {};
\node[node_a_solution] (P21) at (p21) {};
\node[node_a] (P31) at (p31) {};
\node[node_a] (P41) at (p41) {};
\node[node_a] (P51) at (p51) {};

\node[node_regular] (P12) at (p12) {};
\node[node_solution] (P22) at (p22) {};
\node[node_regular] (P32) at (p32) {};
\node[node_regular] (P42) at (p42) {};
\node[node_regular] (P52) at (p52) {};

\node[node_regular] (P13) at (p13) {};
\node[node_solution] (P23) at (p23) {};
\node[node_regular] (P33) at (p33) {};
\node[node_regular] (P43) at (p43) {};
\node[node_regular] (P53) at (p53) {};

\node[node_b] (P14) at (p14) {};
\node[node_b] (P24) at (p24) {};
\node[node_regular] (P34) at (p34) {};
\node[node_regular] (P44) at (p44) {};
\node[node_regular] (P54) at (p54) {};

\draw[edge_regular] (P51) -- (P42);
\draw[edge_regular] (P42) -- (P33);
\draw[edge_regular] (P41) -- (P32);

\draw[edge_php] (P31) -- (P12);
\draw[edge_php] (P11) -- (P22);
\draw[edge_php] (P32) -- (P13);
\draw[edge_php] (P12) -- (P23);
\draw[edge_php] (P33) -- (P14);
\draw[edge_php] (P13) -- (P24);


\coordinate (o11) at (\dx - \spacex, 1*\spacey); \coordinate(o12) at (\dx, 1*\spacey);
\coordinate (o21) at (\dx - \spacex, 2*\spacey); \coordinate(o22) at (\dx, 2*\spacey);
\coordinate (o31) at (\dx - \spacex, 3*\spacey); \coordinate(o32) at (\dx, 3*\spacey);

\coordinate (d11) at ([shift=({-\polygonshift,-\polygonshift})]o11);
\coordinate (d12) at ([shift=({\polygonshift,-\polygonshift})]o12);
\coordinate (d32) at ([shift=({\polygonshift,\polygonshift})]o32);
\coordinate (d31) at ([shift=({-\polygonshift, \polygonshift})]o31);

\node[node_text] at (\dx -.5*\spacex, 1*\spacey + \labelshift) {\textcolor{color_gadget_PHP_label}{$\iphp$}};
\draw[gadget_PHP] (d11) -- (d12) -- (d32) -- (d31) -- cycle;

\node[node_a] (O11) at (o11) {};
\node[node_b] (O12) at (o12) {};
\node[node_a_solution] (O21) at (o21) {};
\node[node_b] (O22) at (o22) {};
\node[node_a] (O31) at (o31) {};
\node[node_regular] (O32) at (o32) {};

\draw[edge_php] (O11) -- (O22);
\draw[edge_php] (O31) -- (O12);


\node[node_text] at (.5*\dx, 2*\spacey) {\Large $\boldsymbol{\leq}$};
\end{tikzpicture}
\caption{The reduction $\iphp\leq\pathiphp_f$ in \cref{lem:inj-path-php}. We build a $\pathiphp_f$ instance by chaining several identical $\iphp$ instances side-by-side. Note that a solution to the $\pathiphp_f$ instance can be directly mapped to one for the original $\iphp$ instance.}
\label{figure:reduction_php_pathphp}
\end{figure}
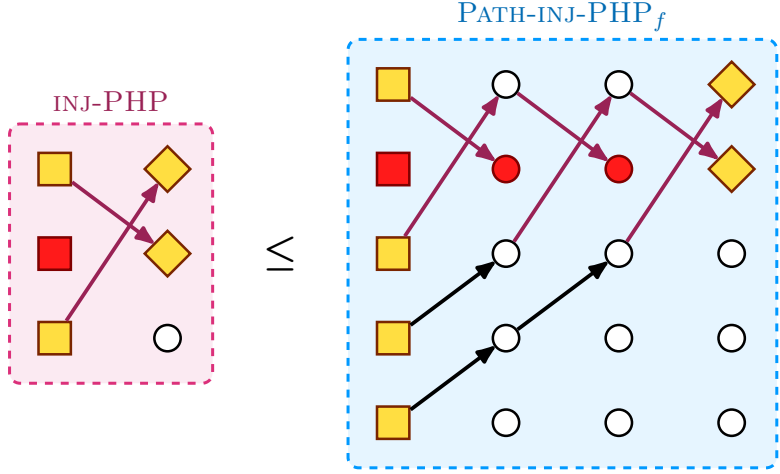

\section{\texorpdfstring{$\SOPL = \PLS\cap\PPADS$}{SOPL = PLS \cap PPADS}}

In this section, we prove \cref{thm:SOPL-collapse}, namely $\SOPL = \PLS\cap\PPADS$. To prove this we provide a reduction from a $\PLS\cap\PPADS$-complete problem to the canonical \SOPL-complete problem.

\begin{lemma}
$\sod\curlywedge\iphp \leq \sopl$.
\label{lem:iphp_or_iter_reduces_to_sopl}
\end{lemma}

\begin{proof}[Proof Sketch]
There are two obstacles to a direct reduction from \sod to \sopl: (i) we can only compute edges in the forward direction (i.e., we only have access to a successor circuit), and (ii) multiple edges can point to the same node.

To resolve the first issue, we modify the original $[N] \times [M]$ grid of the \sod instance by taking $N$ copies of each node. This ensures that there is a separate copy of each node $v$ for each potential predecessor on the previous column. As a result, edges from different predecessor nodes will point to different copies of $v$. This means that predecessor nodes can now also be computed efficiently. Namely, in order to compute the predecessors of the $i$th copy of node $v$, it suffices to check whether in the original \sod instance the $i$th node on the previous column points to $v$. If that is the case, then all copies of this $i$th node are predecessors in the modified instance. Otherwise, there are no predecessors. However, the second issue remains: since we have made $N$ copies of each node, there are also $N$ copies of each predecessor node, and thus $N$ edges pointing to the corresponding copy of $v$. This is not acceptable, since \sopl allows at most one incoming edge.

To overcome the second obstacle, we make use of the following high-level idea: use the $\iphp$ instance (which maps $K+1$ pigeons to $K$ holes) to ``hide'' the fact that multiple edges can point to a single node. Unfortunately, we cannot use the $\iphp$ instance to hide the fact that $N$ paths merge into a single node. But, if we take~$K N$ copies of each original node, instead of just $N$, then we have $KN$ paths and $K$ target nodes. By \cref{lem:inj-path-php}, the $\iphp$ instance can be turned into a $\pathiphp_f$ instance that hides the fact that $KN$ paths merge into $K$ paths. Thus, we replace each original node $v$ of the \sod instance by a gadget that has $KN$ nodes in the left-most column and $K$ nodes in the right-most column, and such that finding the sink of a path inside the gadget requires solving the $\iphp$ instance. Importantly, we only construct the paths inside this gadget when the original node $v$ has a successor in the \sod instance. This ensures that when $v$ is an isolated node, the corresponding gadget does not contain any edges. \Cref{figure:reduction_php_or_iter_sopl} illustrates the construction for $N=3$, $M=4$ and $K=2$.

\begin{figure}[bt]
\centering
\input{figures/reduction_php_or_iter_sopl_2.tex}
\caption{The reduction $\sod\curlywedge\pathiphp\leq\sopl$ in the proof of \Cref{lem:iphp_or_iter_reduces_to_sopl}. Given instances of $\sod$ and $\pathiphp$, we build an $\sopl$ instance whose solutions can be traced back to solutions of $\sod \curlywedge \pathiphp$. To overcome the issue of merging paths in $\sod$, the nodes of the $\sod$ instance are replaced with a copy of a $\pathiphp$ gadget (in blue).
Those gadgets are ultimately built out of the initial $\iphp$ instance (not shown) using \cref{lem:inj-path-php}.}
\label{figure:reduction_php_or_iter_sopl}
\end{figure}

Note that although we might have added many new sources to the graph (which are irrelevant for \sopl), it remains the case that from any sink of the new graph, we can extract either a solution to \sod or to $\iphp$.

In the final construction, edges can indeed be computed in both directions efficiently. Namely, given any node, we can determine in polynomial time if it has an incoming and/or outgoing edge, as well as the identity of the potential predecessor and successor nodes. Here, we crucially use the fact that edges can be computed efficiently in both directions in the $\iphp$ instance.
\end{proof}

\begin{proof}
Let $S$ be an instance of \sod on the grid $[N] \times [M]$. We are also given an instance of $\iphp$ with parameters $(A,B) = (K+1,K)$. Without loss of generality we can assume that $K = N$, because we can easily pad the \sod or $\iphp$ instance with additional rows without changing the set of solutions. By \cref{lem:inj-path-php}, we can reduce this $\iphp$ instance to a $\pathiphp_{t^2}$ instance on the grid $[N^2] \times [M']$ with parameters $(A,B) = (N^2,N)$. Without loss of generality, we can assume that $M' = M$, because we can pad the \sod or $\pathiphp_{t^2}$ instance with additional columns, if needed. This is not important for the reduction, but will be convenient.

We will take $N^2$ copies of each node in the original \sod instance, and make $M$ copies of each column. As a result, our \sopl instance will be defined on the $[N^3] \times [M^2]$ grid. It will be convenient to use some special notation to refer to points in this grid. For $\alpha \in [N^2]$ and $x \in [N]$, we use the notation $(\alpha, x)$ to denote the row $\alpha + (x-1) \cdot N^2 \in [N^3]$. This corresponds to indexing the $\alpha$th copy of row $x$ of the original instance. We also introduce some additional notation to index these $[N^2]$ copies: for $i,j \in [N]$, we let $[i,j] \coloneqq i + (j-1) \cdot N \in [N^2]$. Thus, $([i,j], x)$ denotes the $[i,j]$th copy of row $x$. The ``$[i,j]$'' notation essentially subdivides $[N^2]$ into $N$ blocks containing $N$ values each, which will be useful for routing incoming edges to the correct copy of a node. Using the analogous subdivision also on the columns, the notation $(\alpha,x;k,y) \in [N^2] \times [N] \times [M] \times [M]$ denotes the node $(\alpha + (x-1) \cdot N^2, k + (y-1) \cdot M) \in [N^3] \times [M^2]$. In particular, the notation $([i,j],x;k,y)$ is well-defined.

The circuits $\widehat{S}, \widehat{P}$ of the \sopl instance on $[N^3] \times [M^2]$ are defined as follows:
\begin{align*}
    \widehat{S}([i,j],x;k,y)
    &~\coloneqq~
    \begin{cases}
    ([i,x],S(x,y)) & \text{if } k = M \text{ and } j = 1,\\
    (S'([i,j],k),x) & \text{if } k < M \text{ and } S(x,y) \neq \nul,\\
    \nul & \text{otherwise}
    \end{cases} \\[1em]
    \widehat{P}([i,j],x;k,y)
    &~\coloneqq~
    \begin{cases}
    ([i,1],j) & \text{if } k = 1 \text{ and } y > 1 \text{ and } S(j,y-1) = x,\\
    (P'([i,j],k),x) & \text{if } k > 1 \text{ and } S(x,y) \neq \nul,\\
    \nul & \text{otherwise}
    \end{cases}
\end{align*}
where $(\alpha,z) \in [N^2] \times [N]$ is interpreted as an element in $[N^3]$ as above, and where we use the convention $(\ast,\nul) = (\nul,\ast) = \nul$. Using the fact that $S'$ and $P'$ are consistent, it can be checked that $\widehat{S}$ and $\widehat{P}$ are also consistent.

In order to argue about the correctness of the reduction, consider any sink $([i,j],x;k,y)$ of the \sopl instance. If $2 \leq k \leq M-1$, then it must be that $([i,j],k)$ is a sink of the $\pathiphp_{t^2}$ instance $(S',P')$. If $k = M$ and $j = 1$, then $([i,j],x;k,y)$ cannot be a sink, since $\widehat{P}([i,j],x;k,y) \neq \nul$ implies that $S(x,y) \neq \nul$, and thus $\widehat{S}([i,j],x;k,y) \neq \nul$. If $k = M$ and $j > 1$, then $([i,j],k)$ is an invalid sink on the last column of the $\pathiphp_{t^2}$ instance, and so in particular a solution. If $k = 1$ and $S(x,y) \neq \nul$, then $([i,j],k)$ is a missing source on the first column of the $\pathiphp_{t^2}$ instance, and so again a solution. Finally, if $k = 1$ and $S(x,y) = \nul$, then it must be that $S(j,y-1) = x$ and thus $(x,y)$ is a sink of the original \sod instance, and this is witnessed by the node $(j,y-1)$.
\end{proof}

\section{\texorpdfstring{$\EOPL = \PLS\cap\PPAD$}{EOPL = PLS \cap PPAD}}

In this section, we prove \cref{thm:EOPL-collapse}, namely $\EOPL = \PLS\cap\PPAD$. The equality $\SOPL = \PLS\cap\PPADS$ (\cref{thm:SOPL-collapse}) proved in the previous section, together with the fact that $\PPAD \subseteq \PPADS$, immediately imply that
\[\SOPL\cap\PPAD \,=\, \PLS\cap\PPAD.\]
As a result, in order to prove \cref{thm:EOPL-collapse}, it suffices to give a reduction from an $\SOPL\cap\PPAD$-complete problem to an $\EOPL$-complete problem:

\begin{lemma}
$\sopl \curlywedge \bphp\leq\eopl$.
\label{lem:bphp_or_sopl_reduces_to_eopl}
\end{lemma}

\begin{proof}[Proof Sketch]
A very natural attempt at a reduction from \sopl to \eopl is to try to remove all undistinguished sources, i.e., all sources except the trivial one. Then, clearly, any \eopl-solution would have to be a sink, and thus also a solution to \sopl.

There is a simple trick that \emph{almost} achieves this. First, make a \emph{reversed} copy of the \sopl instance, i.e., reverse the direction of all edges, and the ordering of the potential. Note that sources of the original instance have now become sinks in the reversed copy, and vice versa. Then, for each source node $v$ of the original graph, add an edge pointing from its copy $\overline{v}$ (which is a sink) to $v$.

\begin{figure}[t]
\centering
\newcommand{\gadgetspace}{3}
\newcommand{\labelshift}{-1}

\begin{tikzpicture}[y=-1cm, scale=1]
\coordinate (p11) at (0, 0);				 \coordinate (p12) at (\spacex, 0);				 \coordinate (p13) at (2*\spacex, 0);				\coordinate (p14) at (2*\spacex + \gadgetspace, 0);				\coordinate (p15) at (\gadgetspace + 3*\spacex, 0);				 \coordinate (p16) at (\gadgetspace + 4*\spacex, 0);
\coordinate (p21) at (0, \spacey); 	 \coordinate (p22) at (\spacex, \spacey);	 \coordinate (p23) at (2*\spacex, \spacey); 	\coordinate (p24) at (2*\spacex + \gadgetspace, \spacey);		\coordinate (p25) at (\gadgetspace + 3*\spacex, \spacey);	 \coordinate (p26) at (\gadgetspace + 4*\spacex, \spacey);
\coordinate (p31) at (0, 2*\spacey); \coordinate (p32) at (\spacex, 2*\spacey); \coordinate (p33) at (2*\spacex, 2*\spacey); \coordinate (p34) at (2*\spacex + \gadgetspace, 2*\spacey); \coordinate (p35) at (\gadgetspace + 3*\spacex, 2*\spacey); \coordinate (p36) at (\gadgetspace + 4*\spacex, 2*\spacey);

\coordinate (c11) at ([shift=({-\polygonshift,-\polygonshift})]p11);
\coordinate (c13) at ([shift=({\polygonshift,-\polygonshift})]p13);
\coordinate (c33) at ([shift=({\polygonshift,\polygonshift})]p33);
\coordinate (c31) at ([shift=({-\polygonshift, \polygonshift})]p31);

\coordinate (c14) at ([shift=({-\polygonshift,-\polygonshift})]p14);
\coordinate (c16) at ([shift=({\polygonshift,-\polygonshift})]p16);
\coordinate (c36) at ([shift=({\polygonshift,\polygonshift})]p36);
\coordinate (c34) at ([shift=({-\polygonshift, \polygonshift})]p34);

\coordinate (s1) at (2*\spacex + .5*\gadgetspace, -\spacey);
\coordinate (s2) at (2*\spacex + .5*\gadgetspace, 0);

\draw[gadget_SOPL] (c11) -- (c13) -- (c33) -- (c31) -- cycle;
\draw[gadget_SOPL] (c14) -- (c16) -- (c36) -- (c34) -- cycle;

\small
\node[node_solution] (P11) at (p11) {};
\node[node_regular] (P12) at (p12) {};
\node[node_solution,naive] (P13) at (p13) {\color{white}$\overline{v}_0$};
\node[node_regular,naive] (P14) at (p14) {$v_0$};
\node[node_regular] (P15) at (p15) {};
\node[node_solution] (P16) at (p16) {};
\node[node_solution] (P21) at (p21) {};
\node[node_regular] (P22) at (p22) {};
\node[node_regular] (P23) at (p23) {};
\node[node_regular] (P24) at (p24) {};
\node[node_regular] (P25) at (p25) {};
\node[node_solution] (P26) at (p26) {};
\node[node_regular] (P31) at (p31) {};
\node[node_regular] (P32) at (p32) {};
\node[node_regular] (P33) at (p33) {};
\node[node_regular] (P34) at (p34) {};
\node[node_regular] (P35) at (p35) {};
\node[node_regular] (P36) at (p36) {};

\draw[edge_regular] (P11) -- (P32);
\draw[edge_regular] (P32) -- (P33);
\draw[edge_regular] (P21) -- (P12);
\draw[edge_regular] (P12) -- (P13);

\draw[edge_eopl] (P33) edge[bend left=10] (P34);

\draw[edge_regular] (P14) -- (P15);
\draw[edge_regular] (P15) -- (P26);
\draw[edge_regular] (P34) -- (P35);
\draw[edge_regular] (P35) -- (P16);

\small
\node[node_a, naive] (S1) at (s1) {$u$};
\draw[edge_eopl, rounded corners] (S1) -- (s2) -- (P14);

\node[] at ([shift=({0,\labelshift})]p12) {\textcolor{color_gadget_SOPL_label}{reversed $\sopl$}};
\node[] at ([shift=({0,\labelshift})]p15) {\textcolor{color_gadget_SOPL_label}{original $\sopl$}};
\end{tikzpicture}
\caption{A naive attempt at a reduction $\sopl\leq\eopl$. Although most solutions arise from the sinks of the original $\sopl$ instance, a spurious solution is introduced at $\bar{v}_0$, which does not correspond to any sink of the original instance.
}
\label{figure:reduction_sopl_eopl}
\end{figure}
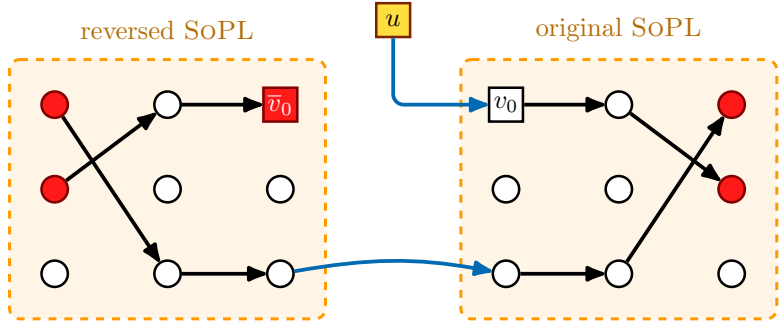

The only problem with this reduction is that we have eliminated \emph{all} sources of the original graph, including the distinguished one. In particular, the distinguished source $v_0$ of the original instance is no longer a source, since there is an edge from its copy $\overline{v}_0$ to $v_0$. As a result, the reduction fails, because the instance of \eopl we have constructed does not have a distinguished source. Furthermore, we cannot hope to turn one of the new sources into a distinguished source, since any such source yields a solution to the original instance (where it is a sink).

In order to address this issue, we add a new node $u$ and select it as our new distinguished source. Clearly, $u$ is a solution of the instance, since it is a distinguished source that is not actually a source, but just an isolated node. Now, imagine that we remove the edge $(\overline{v}_0,v_0)$ and instead introduce an edge $(u,v_0)$; see \cref{figure:reduction_sopl_eopl}. Then, $u$ is no longer a solution, but $\overline{v}_0$ becomes a sink, and thus a solution, instead. In other words, the reduction can pick whether it wants $u$ or $\overline{v}_0$ to be a solution by changing this edge. Of course, in both cases, the resulting instance is very easy to solve, but this minor observation already provides the idea for the next step.

Take $k$ copies of the instance we have constructed (before adding $u$). There are now $k$ copies $v_0^{(1)}, \dots, v_0^{(k)}$ of the original distinguished source, and $k$ copies of the reverse copy $\overline{v}_0^{(1)}, \dots, \overline{v}_0^{(k)}$. Remove the edges $(\overline{v}_0^{(i)}, v_0^{(i)})$ for $i = 1, \dots, k$. If we now introduce the new distinguished source $u$, we have $k+1$ nodes that ``need'' an outgoing edge in order to not be solutions (namely, $u, \overline{v}_0^{(1)}, \dots, \overline{v}_0^{(k)}$) and $k$ nodes that ``need'' an incoming edge (namely, $v_0^{(1)}, \dots, v_0^{(k)}$). Clearly, no matter how we introduce edges here, one of $u, \overline{v}_0^{(1)}, \dots, \overline{v}_0^{(k)}$ will not have an outgoing edge and will be a solution. However, we can use a $\bphp$ instance to make it hard to find such a solution. Let $K$ denote the parameter of the $\bphp$ instance, i.e., $K+1$ points are mapped to $K$ points. Then, we let $k \coloneqq K$ and add edges between $u, \overline{v}_0^{(1)}, \dots, \overline{v}_0^{(k)}$ and $v_0^{(1)}, \dots, v_0^{(k)}$ according to the $\bphp$ instance. An example of the construction is depicted in \Cref{figure:reduction_php_or_sopl_eopl}.

\begin{figure}[th]
\centering
\input{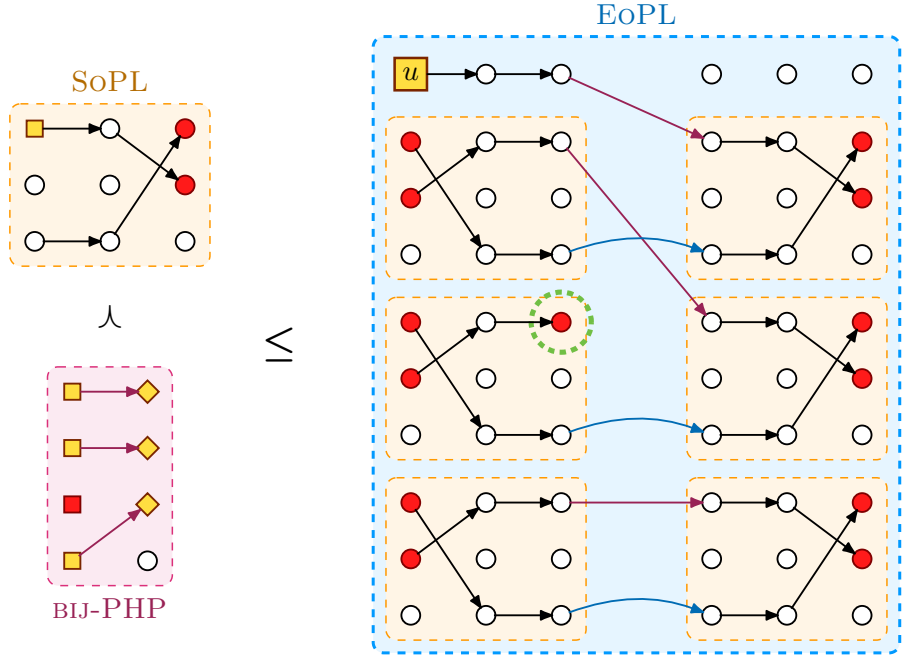}
\caption{The reduction $\sopl\curlywedge\bphp\leq\eopl$ in \cref{lem:bphp_or_sopl_reduces_to_eopl}. The $\bphp$ instance (in pink) connects the newly introduced source $u$ together with the distinguished sources and sinks of the copied $\sopl$ instances. Non-distinguished sources and sinks are connected with blue edges. The node circled in green corresponds to a solution of the $\bphp$ instance.}
\label{figure:reduction_php_or_sopl_eopl}
\end{figure}

Now, it is easy to check that any undistinguished source or any sink of the resulting graph must yield a solution to the $\bphp$ instance or a solution of the \sopl instance. In particular, if $u$ is not a source, then this yields a solution to $\bphp$.
\end{proof}

\begin{proof}
Let $(S,P)$ be an instance of \sopl on the grid $[N] \times [M]$. Without loss of generality, we can assume that all sources occur on the first column, i.e., for any source $(x,y) \in [N] \times [M]$ it holds that $y=1$. Indeed, by appropriately increasing $N$, for each source $(x,y)$ we can add a path that starts on the first column and ends at $(x,y)$, thus effectively ``transferring'' the source to the first column. Let $(S',P')$ be an instance of $\bphp$ on the grid $[K+1] \times [2]$ that maps $K+1$ pigeons to $K$ holes.

We take $K$ copies of the \sopl instance and $K$ copies of the reversed \sopl instance, all together in a single grid. This grid will be of the form $[KN] \times [2M]$. For clarity, we will use the notation $(i,x;y) \in [K] \times [N] \times [2M]$ to denote the element $(x + (i-1) \cdot N,y) \in [KN] \times [2M]$. The $i$th copy of the instance will be embedded in $\{i\} \times [N] \times ([2M] \setminus [M])$, while the $i$th reversed copy will be in $\{i\} \times [N] \times [M]$. Formally, we define new successor and predecessor circuits $\widehat{S}, \widehat{P}$ on $[KN] \times [2M]$ as follows:
\begin{align*}
    \widehat{S}(i,x;y)
    &~\coloneqq~
    \begin{cases}
    (i,P(x,M-y+1)) & \text{if } y \leq M,\\
    (i,S(x,y-M)) & \text{if } y \geq M+1
    \end{cases}\\[1em]
    \widehat{P}(i,x;y)
    &~\coloneqq~
    \begin{cases}
    (i,S(x,M-y+1)) & \text{if } y \leq M,\\
    (i,P(x,y-M)) & \text{if } y \geq M+1
    \end{cases}
\end{align*}
where $(i,z) \in [K] \times [N]$ represents the element $z + (i-1) \cdot N \in [KN]$, and where we use the convention $(i,\nul) = \nul$.

Since $S$ and $P$ are consistent, $\widehat{S}$ and $\widehat{P}$ are also consistent. Note that there are currently no edges between column $M$ and column $M+1$. In the second step of the reduction we add edges between these two columns as follows. For every $i \in [K]$ and $x \in [N] \setminus \{1\}$, if $(i,x;M+1)$ is a source, then we add an edge from $(i,x;M)$ to $(i,x;M+1)$. Note that in that case $(i,x;M)$ was a sink. The case where $x=1$ is handled separately, because it corresponds to nodes that are copies of the distinguished source of the original \sopl instance. For any $i \in [K]$ and for $x=1$, if $S'(i,1) = j \neq \nul$, we add an edge from $(i,1;M)$ to $(j,1;M+1)$. Note that here we also use $P'$ (which is assumed to be consistent with $S'$) to implement this edge in $(\widehat{S},\widehat{P})$.

Finally, we introduce a new special node $u$ on column $M$ which will act as our new distinguished source. If $S'(K+1,1) = j \neq \nul$, then we add an edge from $u$ to $(j,1;M+1)$. By extending the grid to be $[KN+1] \times [2M]$, by renaming nodes and by ``transferring'' the source $u$ to the first column as before, we can ensure that the distinguished source is $(1,1)$.

It is easy to check that the new circuits $\widehat{S}, \widehat{P}$ can be constructed in polynomial time. For the correctness of the reduction, note that any source or sink that occurs on columns $[2M] \setminus \{M,M+1\}$ must correspond to a sink of the original \sopl instance. On the other hand, any source or sink that occurs on column $M$ or $M+1$ must correspond to a solution of the $\bphp$ instance (namely, a pigeon without a hole, or a hole without a pigeon). This completes the reduction.
\end{proof}

\section{Discussion}\label{sec:discussion}

As mentioned in the introduction, it remains open whether $\UEOPL \overset{?}{=} \EOPL$. Separating the two classes in the black-box model would be an important first step towards pinning down the complexity of the various natural problems contained in \UEOPL, since it would provide strong evidence that these problems are unlikely to be complete for $\PLS\cap\PPAD$.

The techniques developed in this paper do not seem to yield any other major class collapse. Indeed, our reductions are all black-box, and the main classes are known to be distinct in that model~\cite{Beame1998,Morioka2001,Buresh2004,Goos2022}.

In the remainder of this section we briefly present some observations about the path pigeonhole problems, as well as a further consequence of our reduction techniques: a version of $\sod$ where paths are not allowed to merge turns out to be $\PLS\cap\PPP$-complete.

\paragraph{Path-Pigeonhole problems.}
\cref{lem:inj-path-php} in particular establishes that $\pathiphp_f$ is \PPADS-hard. Membership in \PPADS can be shown by reducing to $\iphp$ using a construction similar to the reduction from $\eol$ to $\bphp$ in the proof of \cref{lem:PPAD-as-PHP}.

The statement of \cref{lem:inj-path-php} also holds for $\bphp\leq \pathbphp_f$, and the proof is essentially the same. This shows that $\pathbphp_f$ is \PPAD-hard. However, it is unclear whether $\pathbphp_f$ lies in \PPAD. Indeed, using the same idea as for $\pathiphp_f \leq \iphp$ yields an instance with $A \gg B$, and we cannot increase $B$ artificially here (whereas this is possible in $\iphp$). Another way to state this is to say that we can reduce $\pathbphp_f$ to an instance of $\eol$ that has many distinguished source nodes, instead of just one. It is known that $\eol$ with a polynomial number of distinguished sources remains \PPAD-complete~\cite{Goldberg2021}, but in general we will obtain an exponential number of such sources here.

\paragraph{Extending the $\grid$ problem to capture \PPP.}
The canonical \PPP-complete problem is $\pigeoncircuit$~\cite{Papadimitriou1994}: given a circuit mapping $N$ pigeons to $N-1$ holes, find a \emph{collision}, i.e., two pigeons that are mapped to the same hole. Importantly, unlike in $\iphp$ or $\bphp$, we are not given a circuit to compute the mapping in the other direction, i.e., from holes to pigeons. In order to capture this problem, we extend the definition of $\grid$ by introducing an additional parameter bit $c \in \{0,1\}$, which stands for \emph{collision}. We also introduce a new solution type:
\begin{enumerate}
    \item[5.] If $r=0$ and $c=1$: $x_1,x_2 \in [N]$ and $y \in [M-1]$ such that \\
    $x_1 \neq x_2$ and $S(x_1,y) = S(x_2,y) \neq \nul$, \hfill \emph{(pigeon collision/merging)}
\end{enumerate}
Furthermore, the syntactic condition ``If $r=0$, then $b=0$ and $B=0$'' is replaced by the condition:
\begin{itemize}
    \item If $r=0$, then $b=0$. If $r=0$ and $c=0$, then $B=0$.
\end{itemize}
The \PPP-complete problem $\pigeoncircuit$ is then obtained by setting $r=0, c=1, M=2, N=A=B+1$. In fact, $\grid$ remains in \PPP even if we just set $r=0, c=1$ and leave the other parameters unfixed. This can be shown by using a construction similar to the reduction from $\eol$ to $\bphp$ in the proof of \cref{lem:PPAD-as-PHP}.

\paragraph{$\sod$ without merging.}
What is the complexity of $\sod$ if paths are not allowed to merge? In other words, what is the complexity of the $\grid$ problem with parameters $r=0, c=1, A=1, B=0$? Clearly, this restricted version still lies in \PLS, and by the previous paragraph it also lies in \PPP. Using the ideas developed in this paper, it can be shown that the problem is in fact $\PLS\cap\PPP$-complete. To see this, note that using the simple construction in the proof of \cref{lem:inj-path-php} we can reduce $\pigeoncircuit$ to a path-version of the problem where $f(T)$ pigeons are mapped to $T$ holes. Then, the construction in the proof of \cref{lem:iphp_or_iter_reduces_to_sopl} can be used to reduce $\sod \curlywedge \pigeoncircuit$ to $\sod$ without merging.

\bigskip
\subsection*{Acknowledgements}

We thank Aviad Rubinstein for his many questions during e-mail correspondence, and the anonymous reviewers for their suggestions that helped improve the presentation of the paper.

\bigskip

\small

\DeclareUrlCommand{\Doi}{\urlstyle{sf}}
\renewcommand{\path}[1]{\small\Doi{#1}}
\renewcommand{\url}[1]{\href{#1}{\small\Doi{#1}}}
\bibliographystyle{alphaurl}
\bibliography{eopl-refs}

\newcommand{\etalchar}[1]{$^{#1}$}
\begin{thebibliography}{DGK{\etalchar{+}}17}

\bibitem[BCE{\etalchar{+}}98]{Beame1998}
Paul Beame, Stephen Cook, Jeff Edmonds, Russell Impagliazzo, and Toniann
  Pitassi.
\newblock The relative complexity of {NP} search problems.
\newblock {\em Journal of Computer and System Sciences}, 57(1):3--19, 1998.
\newblock \href {https://doi.org/10.1006/jcss.1998.1575}
  {\path{doi:10.1006/jcss.1998.1575}}.

\bibitem[BM04]{Buresh2004}
Joshua Buresh{-}Oppenheim and Tsuyoshi Morioka.
\newblock Relativized {NP} search problems and propositional proof systems.
\newblock In {\em Proceedings of the 19th IEEE Conference on Computational
  Complexity (CCC)}, pages 54--67, 2004.
\newblock \href {https://doi.org/10.1109/CCC.2004.1313795}
  {\path{doi:10.1109/CCC.2004.1313795}}.

\bibitem[CCM20]{Chiu2020}
Man-Kwun Chiu, Aruni Choudhary, and Wolfgang Mulzer.
\newblock {Computational Complexity of the $\alpha$-Ham-Sandwich Problem}.
\newblock In {\em 47th International Colloquium on Automata, Languages, and
  Programming (ICALP)}, pages 31:1--31:18, 2020.
\newblock \href {https://doi.org/10.4230/LIPIcs.ICALP.2020.31}
  {\path{doi:10.4230/LIPIcs.ICALP.2020.31}}.

\bibitem[Das19]{Daskalakis2019}
Constantinos Daskalakis.
\newblock Equilibria, fixed points, and computational complexity.
\newblock In {\em Proceedings of the International Congress of Mathematicians
  (ICM)}. World Scientific, 2019.
\newblock \href {https://doi.org/10.1142/9789813272880_0009}
  {\path{doi:10.1142/9789813272880_0009}}.

\bibitem[DGK{\etalchar{+}}17]{Dohrau2017}
J{\'{e}}r{\^{o}}me Dohrau, Bernd G\"artner, Manuel Kohler, Ji{\v{r}}{\'{\i}}
  Matou{\v{s}}ek, and Emo Welzl.
\newblock {ARRIVAL}: A zero-player graph game in {NP} $\cap$ {coNP}.
\newblock In {\em A Journey Through Discrete Mathematics}, pages 367--374.
  Springer, 2017.
\newblock \href {https://doi.org/10.1007/978-3-319-44479-6_14}
  {\path{doi:10.1007/978-3-319-44479-6_14}}.

\bibitem[DGP09]{Daskalakis2009}
Constantinos Daskalakis, Paul Goldberg, and Christos Papadimitriou.
\newblock The complexity of computing a {N}ash equilibrium.
\newblock {\em SIAM Journal on Computing}, 39(1):195--259, 2009.
\newblock \href {https://doi.org/10.1137/070699652}
  {\path{doi:10.1137/070699652}}.

\bibitem[DP11]{Daskalakis2011}
Constantinos Daskalakis and Christos Papadimitriou.
\newblock Continuous local search.
\newblock In {\em Proceedings of the 22nd Symposium on Discrete Algorithms
  (SODA)}, pages 790--804, January 2011.
\newblock \href {https://doi.org/10.1137/1.9781611973082.62}
  {\path{doi:10.1137/1.9781611973082.62}}.

\bibitem[EPRY20]{Etessami2020}
Kousha Etessami, Christos Papadimitriou, Aviad Rubinstein, and Mihalis
  Yannakakis.
\newblock {T}arski’s theorem, supermodular games, and the complexity of
  equilibria.
\newblock In {\em Proceedings of the 11th Innovations in Theoretical Computer
  Science Conference (ITCS)}, volume 151, pages 18:1--18:19, 2020.
\newblock \href {https://doi.org/10.4230/LIPIcs.ITCS.2020.18}
  {\path{doi:10.4230/LIPIcs.ITCS.2020.18}}.

\bibitem[FGHS21]{Fearnley2021}
John Fearnley, Paul Goldberg, Alexandros Hollender, and Rahul Savani.
\newblock The complexity of gradient descent: {CLS} $=$ {PPAD} $\cap$ {PLS}.
\newblock In {\em Proceedings of the 53rd Symposium on Theory of Computing
  (STOC)}, pages 46--59, 2021.
\newblock \href {https://doi.org/10.1145/3406325.3451052}
  {\path{doi:10.1145/3406325.3451052}}.

\bibitem[FGMS20]{Fearnley2020}
John Fearnley, Spencer Gordon, Ruta Mehta, and Rahul Savani.
\newblock Unique end of potential line.
\newblock {\em Journal of Computer and System Sciences}, 114:1--35, 2020.
\newblock \href {https://doi.org/10.1016/j.jcss.2020.05.007}
  {\path{doi:10.1016/j.jcss.2020.05.007}}.

\bibitem[GH21]{Goldberg2021}
Paul Goldberg and Alexandros Hollender.
\newblock The {H}airy {B}all problem is {PPAD}-complete.
\newblock {\em Journal of Computer and System Sciences}, 122:34--62, 2021.
\newblock \href {https://doi.org/10.1016/j.jcss.2021.05.004}
  {\path{doi:10.1016/j.jcss.2021.05.004}}.

\bibitem[GHH{\etalchar{+}}18]{Gaertner2018}
Bernd G{\"a}rtner, Thomas~Dueholm Hansen, Pavel Hub{\'a}cek, Karel Kr{\'a}l,
  Hagar Mosaad, and Veronika Sl{\'i}vov{\'a}.
\newblock {ARRIVAL}: Next stop in {CLS}.
\newblock In {\em 45th International Colloquium on Automata, Languages, and
  Programming (ICALP)}, volume 107, pages 60:1--60:13. Schloss Dagstuhl, 2018.
\newblock \href {https://doi.org/10.4230/LIPICS.ICALP.2018.60}
  {\path{doi:10.4230/LIPICS.ICALP.2018.60}}.

\bibitem[GHJ{\etalchar{+}}22]{Goos2022}
Mika G\"o\"os, Alexandros Hollender, Siddhartha Jain, Gilbert Maystre, William
  Pires, Robert Robere, and Ran Tao.
\newblock Separations in proof complexity and {TFNP}.
\newblock Technical Report TR22-058, Electronic Colloquium on Computational
  Complexity (ECCC), 2022.
\newblock URL: \url{https://eccc.weizmann.ac.il/report/2022/058/}.

\bibitem[GKRS18]{Goos2018}
Mika G{\"o}{\"o}s, Pritish Kamath, Robert Robere, and Dmitry Sokolov.
\newblock Adventures in monotone complexity and {TFNP}.
\newblock In {\em Proceedings of the 10th Innovations in Theoretical Computer
  Science Conference (ITCS)}, volume 124, pages 38:1--38:19, 2018.
\newblock \href {https://doi.org/10.4230/LIPIcs.ITCS.2019.38}
  {\path{doi:10.4230/LIPIcs.ITCS.2019.38}}.

\bibitem[Hol21]{Hollender2021}
Alexandros Hollender.
\newblock {\em Structural results for total search complexity classes with
  applications to game theory and optimisation}.
\newblock PhD thesis, University of Oxford, 2021.
\newblock URL:
  \url{https://ora.ox.ac.uk/objects/uuid:67e2d80b-76bf-4b49-9b7d-8bbd91633dd7}.

\bibitem[HY20]{Hubacek2020}
Pavel Hub{\'{a}}{\v{c}}ek and Eylon Yogev.
\newblock Hardness of continuous local search: Query complexity and
  cryptographic lower bounds.
\newblock {\em {SIAM} Journal on Computing}, 49(6):1128--1172, 2020.
\newblock \href {https://doi.org/10.1137/17m1118014}
  {\path{doi:10.1137/17m1118014}}.

\bibitem[Ish21]{Ishizuka2021}
Takashi Ishizuka.
\newblock The complexity of the parity argument with potential.
\newblock {\em Journal of Computer and System Sciences}, 120:14--41, September
  2021.
\newblock \href {https://doi.org/10.1016/j.jcss.2021.03.004}
  {\path{doi:10.1016/j.jcss.2021.03.004}}.

\bibitem[JPY88]{Johnson1988}
David Johnson, Christos Papadimitriou, and Mihalis Yannakakis.
\newblock How easy is local search?
\newblock {\em Journal of Computer and System Sciences}, 37(1):79--100, August
  1988.
\newblock \href {https://doi.org/10.1016/0022-0000(88)90046-3}
  {\path{doi:10.1016/0022-0000(88)90046-3}}.

\bibitem[MMSS17]{Meunier2017}
Fr{\'e}d{\'e}ric Meunier, Wolfgang Mulzer, Pauline Sarrabezolles, and Yannik
  Stein.
\newblock {The rainbow at the end of the line—a {PPAD} formulation of the
  colorful {C}arath{\'e}odory theorem with applications}.
\newblock In {\em Proceedings of the 28th ACM-SIAM Symposium on Discrete
  Algorithms (SODA)}, pages 1342--1351, 2017.
\newblock \href {https://doi.org/10.1137/1.9781611974782.87}
  {\path{doi:10.1137/1.9781611974782.87}}.

\bibitem[Mor01]{Morioka2001}
Tsuyoshi Morioka.
\newblock Classification of search problems and their definability in bounded
  arithmetic.
\newblock Master's thesis, University of Toronto, 2001.
\newblock URL:
  \url{https://www.collectionscanada.ca/obj/s4/f2/dsk3/ftp04/MQ58775.pdf}.

\bibitem[MP91]{Megiddo1991}
Nimrod Megiddo and Christos Papadimitriou.
\newblock On total functions, existence theorems and computational complexity.
\newblock {\em Theoretical Computer Science}, 81(2):317--324, April 1991.
\newblock \href {https://doi.org/10.1016/0304-3975(91)90200-L}
  {\path{doi:10.1016/0304-3975(91)90200-L}}.

\bibitem[Pap94]{Papadimitriou1994}
Christos Papadimitriou.
\newblock On the complexity of the parity argument and other inefficient proofs
  of existence.
\newblock {\em Journal of Computer and System Sciences}, 48(3):498--532, June
  1994.
\newblock \href {https://doi.org/10.1016/s0022-0000(05)80063-7}
  {\path{doi:10.1016/s0022-0000(05)80063-7}}.

\bibitem[Sch21]{Schnider2021}
Patrick Schnider.
\newblock {The Complexity of Sharing a Pizza}.
\newblock In {\em 32nd International Symposium on Algorithms and Computation
  (ISAAC)}, pages 13:1--13:15, 2021.
\newblock \href {https://doi.org/10.4230/LIPIcs.ISAAC.2021.13}
  {\path{doi:10.4230/LIPIcs.ISAAC.2021.13}}.

\end{thebibliography}

\end{document}